\documentclass[conference]{IEEEtran}
\IEEEoverridecommandlockouts
\usepackage{cite}
\usepackage{amsmath,amssymb,amsfonts}
\usepackage{algorithmic}
\usepackage{graphicx}
\usepackage{textcomp}
\usepackage{xcolor}
\usepackage{booktabs} 
\usepackage{stfloats}
\usepackage{subfigure}
\usepackage{float}
\usepackage{bm}
\usepackage{bbding}
\usepackage{array}
\usepackage{enumitem} 
\usepackage{amsthm}
\usepackage[linesnumbered,ruled,vlined]{algorithm2e}
\newtheorem{definition}{Definition}
\newtheorem{theorem}{Theorem}
\def\BibTeX{{\rm B\kern-.05em{\sc i\kern-.025em b}\kern-.08em
    T\kern-.1667em\lower.7ex\hbox{E}\kern-.125emX}}
\begin{document}

\title{NETR-Tree: An Efficient Framework for Social-Based Time-Aware Spatial Keyword Query
\thanks{This work was supported by the National Key R\&D Program of China [2020YFB1707903]; the National Natural Science Foundation of China [61872238, 61972254], Shanghai Municipal Science and Technology Major Project [2021SHZDZX0102], the Tencent Marketing Solution Rhino-Bird Focused Research Program [FR202001], the CCF-Tencent Open Fund [RAGR20200105], and the Huawei Cloud [TC20201127009]. We thank Nianzu Yang and Zhixian Yang for their contributions to this work. \IEEEauthorrefmark{1} Xiaofeng Gao is the corresponding author.
}}

\author{
\IEEEauthorblockN{Xiuqi Huang, Yuanning Gao, Xiaofeng Gao\IEEEauthorrefmark{1}, Guihai Chen}
\IEEEauthorblockA{MoE Key Lab of Artificial Intelligence, Department of Computer Science and Engineering\\Shanghai Jiao Tong University, Shanghai, China\\Email: {\{huangxiuqi, gyuanning\}@sjtu.edu.cn, \{gao-xf, gchen\}@cs.sjtu.edu.cn}}
}

\maketitle

\begin{abstract}
The development of global positioning system stimulates the popularity of location-based social network (LBSN) services. With a large volume of data containing locations, texts, check-in information, and social relationships, spatial keyword queries in LBSNs have become increasingly complex. In this paper, we identify and solve the \emph{Social-based Time-aware Spatial Keyword Query} (STSKQ) that returns the top-k objects by considering geo-spatial score, keywords similarity, visiting time score, and social relationship effect. To tackle STSKQ, we propose a two-layer hybrid index structure called \emph{Network Embedding Time-aware R-tree} (NETR-Tree). In the user layer, we exploit the network embedding strategy to measure the relationship effect in users' relationship network. In the location layer,  we build a Time-aware R-tree (TR-tree) considered spatial objects' spatio-temporal check-in information, and present a corresponding query processing algorithm. Finally, extensive experiments on two different real-life LBSNs demonstrate the effectiveness and efficiency of our methods, compared with existing state-of-the-art methods.
\end{abstract}

\begin{IEEEkeywords}
Location-based Social Network Service, Top-k Spatial Keyword Query, Network Embedding, Time-aware Query Processing
\end{IEEEkeywords}

\section{Introduction}
Due to the booming popularity of social media and the advance in geo-positioning technology, location-based social networks (LBSNs) have been proliferating in recent years. LBSN Services, such as \emph{Foursquare, Yelp}, have huge amounts of data, consisting of spatial locations, texts, check-in information, and social relationship. 

In LBSNs, users' visiting to spatial places may further be shaped by their social relationships~\cite{cho2011friendship}, as they are more likely to visit places that their friends and people having similar preferences to them visited in the past. In addition, according to the users' check-in records, we can know that different spatial places have different suitable visiting time for visitors. For instance, bars and nightclubs may be less attractive during daytime, and conversely some art museums do not open at night. In particular,~\cite{xiangguo2020social} points that information in social networks can become an important basis for analyzing user preferences in spatial keyword query.

With a rich source of spatio-temporal data and social relationships in LBSNs, we can obtain the true needs of the users more precisely through analysis of these extra information. However, spatial keyword query becomes increasingly complex when taking temporal information and social relationships into consideration at the same time. Therefore, a framework that can aggregate spatial, temporal and social information well is urgently needed to tackle with Social-based Time-aware Spatial Keyword Query (STSKQ).

Recently, various approaches~\cite{li2011ir,liu2015linq,wu2012framework} have been designed to support spatial keyword query, and these approaches mainly focus on keywords and spatial location. As~\cite{chen2017time,cho2011friendship} state, the spatial keyword query cannot satisfy users' requirements if temporal and social information are neglected.~\cite{xiangguo2020social} proposes a hybrid index structure named SAIR-Tree that considers the attributes of social, spatial, and textual information with an approximate algorithm and an exact algorithm. Nevertheless, SAIR-Tree still ignores the temporal information in location-based social networks (LBSNs). Although~\cite{vi2012spatial} seems to take both temporal information and social information into account,~\cite{vi2012spatial} focuses on the real-time service like \emph{Twitter}. In their model, fresher tweets written by users that have more followers have priority. Besides, the number of replies also contributes to the priority. Therefore, their model does not make full use of social network and is not suitable for the services like \emph{Foursquare}.

In this paper, we explore the \emph{Social-based Time-aware Spatial Keyword Query} (STSKQ), which returns a set of top-$k$ objects taking geo-spatial, textual, temporal, social score into consideration. In view of the limitations of the existing frameworks, we design a novel spatial keyword query index, \emph{Network Embedding Time-aware R-tree} (NETR-Tree), and its corresponding query processing algorithms. More specifically, we propose a neighbor selection method which preserves both the local and global network structures based on historical check-in records, and then tackle users' social relationship with their neighbors using the network representation learning (NRL) approach \cite{tang2015line}, i.e., network embedding. We leverage it to learn the structural information of users' representation network. Based on the embedding strategy, the effect of social relationship is measured by the similarity of embedding vectors between users and their neighbors with all users' check-in records. Moreover, to further perform temporal analysis, we split every day hourly inspired by~\cite{wen2015kstr}. Then, we extract check-in time distribution for objects, to measure the visiting time scores for objects at different time slots.

To sum up, the contributions of this paper are three-fold:

\begin{itemize}
\item We formulate the problem of Social-based Time-aware Spatial Keyword Query (STSKQ), which takes geo-spatial score, keywords similarity, visiting time score, and social relationship effect into consideration at the same time to return expected top-k results. 

\item We design a hybrid index structure, i.e., NETR-Tree that exploits network embedding and corresponding efficient pruning strategies to tackle STSKQ. We present a theoretical analysis of NETR-Tree's time complexity, proving that although aggregating extra temporal information and social relationships compared to the basic spatial keyword query, NETR-Tree still holds high efficiency.

\item We conduct extensive experiments to verify the validity and efficiency of the proposed method. Results on two real-world benchmark datasets demonstrate show that our framework outperform the state-of-the-art algorithms for processing STSKQ by a notable margin.
\end{itemize}

The rest of this paper is organized as follows. Sec.~\ref{related} reviews related work. Sec.~\ref{formulation} formulates the problem of STSKQ. We elaborate the NETR-Tree in Sec.~\ref{NETR}. The query processing algorithm based on the NETR-Tree and its time complexity are introduced in Sec.~\ref{query_processing}. In Sec.~\ref{Experiment}, we propose three baseline algorithms and show the experimental results. Finally, we conclude the paper in Sec.~\ref{conclusion}.

\section{Related Work} \label{related}
In this section, we overview the existing techniques for the STSKQ problem, including spatial keyword queries, social-aware spatial keyword queries, social network embedding, and time-aware retrieval.

\subsection{Spatial Keyword Queries}
Recently, spatial keyword queries have been gaining a lot of attention \cite{li2011ir,wu2012framework,chen2021location}. \cite{cao2012spatial} presents a survey for various types of functionality as well as corresponding ideas on spatial keyword query. \cite{chen2013spatial} gives a comprehensive experimental evaluation for different spatial keyword query indices and query processing techniques. Meanwhile, there also exist many methods solving the variants of spatial keyword query, such as collective spatial keyword query~\cite{gao2016efficient,yu2019col}, attribute-aware spatial keyword query~\cite{liu2015linq}, spatial keyword query over streaming data~\cite{wang2017top}, personalized and approximated spatial keyword query approach~\cite{2020personal}, popularity-based top-k spatial keyword query~\cite{2019popularity}, etc. All the above methods only focus on the distance constraint from users‘ spatial information. The sights of social information and temporal information are lost. Hence, these methods cannot effectively solve the STSKQ problem.

\subsection{Soaicl-Aware Spatial Keyword Queries}
With the dramatic increase in social data, many studies have begun to consider the social influence and evolve social-aware SKQ problem. \cite{cozza2013spatio,wu2014social,ahuja2015geo} enriches the semantics of the conventional spatial keyword query by introducing social relevance attributes. \cite{jin2020efficient} attempt to handle batch processing on multiple reverse geo-social keyword queries. \cite{WANG2021,xiangguo2020social} further satisfy the query needs of user groups under spatial proximity and social relevance. \cite{zhao2020efficiently} considers not only the relevance but also the diversity of the result. In addition, there are some studies~\cite{hoang2016unified,chen2017time,chen2021time} combine keywords with spatial and temporal constraints, but ignoring the social relationship. Although various types of SKQ problems have been studied on LBSNs, the existing techniques are not applicable to the STSKQ that aims at finding top-k results under spatial, temporal and social information.

\subsection{Social-Based Network Embedding}

In recent years, neural representation learning in language modeling~\cite{mikolov2013efficient} has made major strides. Lots of embedding learning models have been proposed to learn the embedding vectors of nodes by predicting nodes' neighborhoods. DeepWalk~\cite{perozzi2014deepwalk} exploits the random walk algorithm to generate sequences of instances to obtain the embedding result vectors of nodes. LINE~\cite{tang2015line} is learned from a large-scale information network embedding using the edge-sampling algorithm to improve the effectiveness and gain local relationship influence.

As several works analysis~\cite{cho2011friendship,khani2018context,sohail2018social}, it is notable that a user's interest and behavior often correlate to their friends. However, those papers neglect that in the enormous social network, there are many similar users that have not become friends. Those unacquainted users can also contribute to spatial keyword queries of the target user. In this paper, we show that the network embedding strategy with an elaborated neighbor selection method can be well adopted for STSKQ.

\subsection{Time-Aware Retrieval}

As the factor of time has been gaining increasing importance within search contexts, time-aware retrieval has received much attention from researchers. However, time-aware retrieval in previous research puts emphasis on the spatial, textual, and temporal information, which cannot handle social-based query efficiently. TA-Tree~\cite{chen2017time} proposed a feasible solution that takes a query with visiting time interval. To enable time-aware retrieval, TA-Tree measured visiting time score by the intervals' overlap between the query and spatial objects. In the meantime, there also exist other kinds of time-aware criteria. For instance, \cite{chen2015temporal} utilized exponential time decay function to measure the recency over a stream of geo-textual objects such as tweets. In~\cite{wen2015kstr}, the visiting probability of a point of interest (POI) is defined as a time-aware criterion in the problem of travel route recommendation. It is worth mentioning that there is still a gigantic gap between these criteria and the problem of STSKQ.

\section{Problem Formulation} \label{formulation}

In this section, we formally present concise definitions of STSKQ. Consider a spatial objects dataset $D = \{ o_1, o_2, o_3, \dots \}$. An object $o$ in $D$ is denoted as a tuple $\langle o.l, o.W, o.T \rangle$, where $o.l$ is a spatial location composed of latitude and longitude, $o.W$ is a set of keywords, and $o.T$ 
is the check-in time distribution set for $o$, which is defined as follow.

\begin{definition}[Check-in Time Distribution]\label{def:1}
Given a spatial object $o$, its check-in time distribution is denoted as $o.T$. According to~\cite{wen2015kstr}, we split day time into hourly-based time intervals denoted as $\Gamma$. We define the probability of an object $o$ to be visited during the time interval $\tau$ as $C(o,\tau)/C_{total}(o)$, where $C(o,\tau)$ is the number of check-ins recorded in $o$ during time interval $\tau$, and $C_{total}(o)$ is the total number of check-ins in $o$. Eqn.~\eqref{eq:time_distribution_o} describes the formulation of $o.T$.
\begin{equation} \label{eq:time_distribution_o}
  o.T = \bigcup_{\tau \in \Gamma }\frac{C(o,\tau)}{C_{total}(o)}
\end{equation}

\end{definition}

A time-aware query $q$ is represented as a tuple $\langle q.u, q.l, q.W, q.t\rangle$, where  $q.u, q.l $, and $q.W$ represent a user, the location of the user, and
a set of required keywords respectively, and $q.t$ is a query time stamp at which the query is issued. Next, we formally define social relationship network in LBSN.
\begin{definition}[Social Relationship Network]\label{def:2}
A social relationship network is defined as an unweighted and undirected graph $G = (U,E)$, where $U$ is the set of vertices, each representing a user in LBSN, and $E$ is the set of edges between the vertices, each representing the relationship between two users.
\end{definition}

Accordingly, users are represented as a tuple $\langle \emph{\textbf{V}}, Nrs, C\rangle$, where $\emph{\textbf{V}}$ represents the embedding vectors obtained by network embedding strategy, i.e., LINE~\cite{tang2015line}. For a user $u$, $Nrs(u)$ is the set of $u$'s neighbors and $C(u)$ is a set of $u$'s check-in records numbers. Then we define STSKQ formally.

\begin{definition}[Social-based Time-aware Spatial Keyword Query]\label{def:3}
Given a spatial objects set $D$, a user $u$ in LBSN, and social-based time-aware spatial keyword query $q$ issued by $u$, the query returns a result set $Top_k(u,q)$, where $Top_k(u,q) \subset D$, $\left|Top_k(u,q)\right| = k$, and $\forall o_i, o_j$: $o_i \in Top_k(u,q), o_j \in D - Top_k(u,q)$, it holds that $F(u, q, o_i) \geq F(u, q, o_j)$, where $F$ is the score function.
\end{definition}

Note that, $F(u, q, o)$ in Def. \ref{def:3} is composed of four aspects, including geo-spatial score, keywords similarity, visiting time score, and social relationship effect. $F(u, q, o)$ will be further discussed in the Sec.~\ref{NETR} and Sec.~\ref{query_processing}. Table~\ref{tab:notation} lists the notations and their definitions.

\begin{table}
\centering
\caption{List of Notations}
\label{tab:notation}
      \begin{tabular}{cc}
        \hline
        \textbf{Notation} & \textbf{Definition} \\
        \hline
        $D$ & the spatial objects set, $D = \{ o_1, o_2, o_3, \dots \}$\\
        $o$ & the objects in $D$, $\langle o.l, o.W, o.T \rangle$\\
        $o.l$ & the spatial location\\
        $o.W$ & the set of keywords\\
        $o.T$ & the check-in time distribution\\
        $u$ & the users, $\langle \emph{\textbf{V}}, Nrs, C\rangle$\\
        $\emph{\textbf{V}}$ & the represented vectors of users\\
        $Nrs(u)$ & the set of $u$'s neighbors\\
        $C(u)$ & the set of $u$'s check-in records numbers\\
        $q$ & the social-based time-aware spatial keyword query\\
        $Top_k(u,q)$ & the result set of the query $q$ issued by $u$\\
        $F$ & the ranking score function, $F(u, q, o)$ or $F(u, q, \eta)$\\
        $\eta$ & the node in TR-tree\\
        $N$ & the non-leaf node in TR-tree\\
		$N.rec$ & the minimum bounding rectangle of TR-tree\\
		$N.W$ & \begin{tabular}[c]{@{}c@{}}the TF-IDF weights set of $N$'s\\  descendant nodes' keywords\end{tabular}\\
		$N.T$ & \begin{tabular}[c]{@{}c@{}}the maximal check-in distribution for each time \\ interval of $N$'s descendant nodes\end{tabular}\\
		$N.cEntropy$ & the category entropy \\
		$F_t$ & the visiting time score function, $F_t(\eta,t)$ \\
		$F_s$ & the social relationship effect function, $F_s(u,\eta)$ \\
		$F_g$ & the geo-spatial score function, $F_g(q,\eta)$ \\
		$F_k$ & the keywords similarity function, $F_k(q,\eta)$ \\
        \hline
    \end{tabular}
\end{table}

\section{NETR-Tree}\label{NETR}

\begin{figure*}
  \centering
  \includegraphics[width=0.95\textwidth]{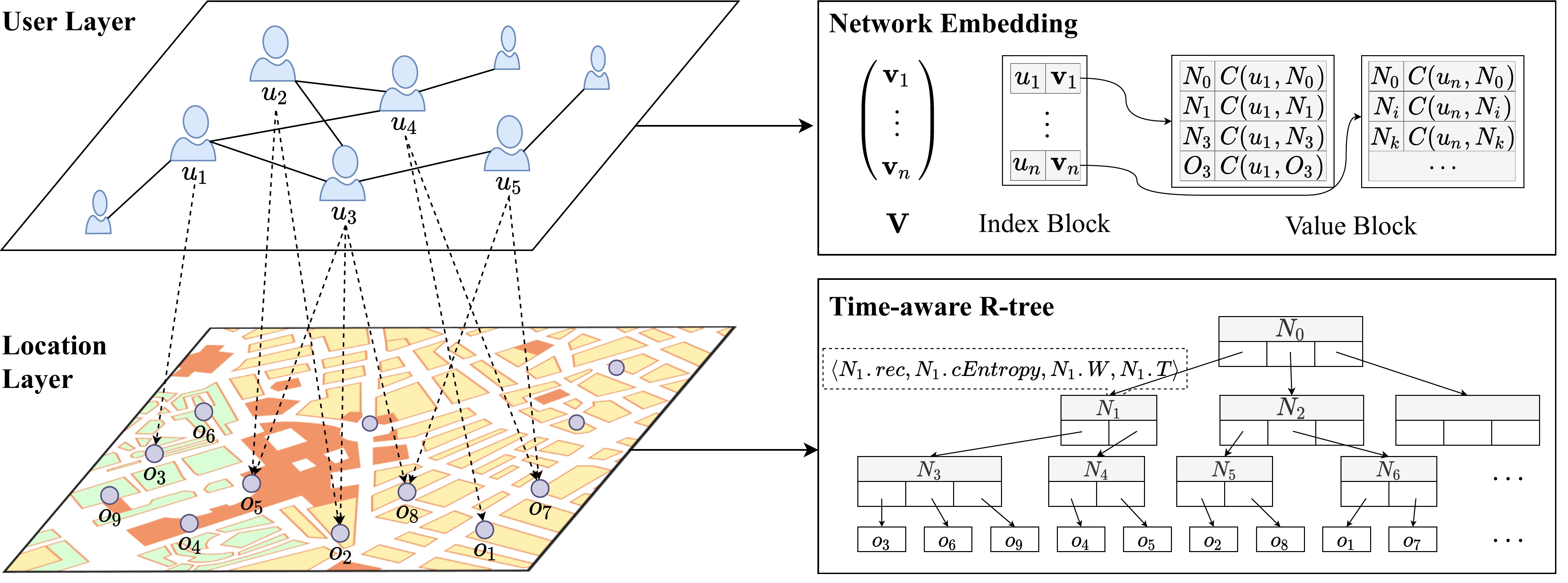}
  \caption{Overview of NETR-Tree}
  \label{fig:overview_NETR-Tree}
\end{figure*}

An overview of the proposed index, NETR-Tree (Network Embedding Time-aware R-tree), is shown in Fig.~\ref{fig:overview_NETR-Tree}. Taking the input from both the user layer and the location layer, we construct the index consisting of two parts: network embedding and Time-aware R-tree (TR-tree). Users and their relationship form a social relationship network. Hence, in the network embedding part, NETR-Tree embeds users into embedding vectors on the basis of their social relationship network structure. Supported by a novel neighbor selection method, for a user $u$, the social effect of $u$'s neighbors can be calculated by the similarity between their embedding vectors and neighbors' historical check-ins. In TR-tree part, each internal node records spatio-textual information together with check-in time distribution. According to the query tuple, TR-tree prunes out spatio-textually and temporally irrelevant objects during the query processing. Therefore, NETR-Tree efficiently tackles the problem of STSKQ.

\subsection{TR-Tree Structure}

As shown in the bottom right of Fig.~\ref{fig:overview_NETR-Tree}, in order to process temporal information, i.e., users' visiting time and objects' check-in time distribution in spatial keyword queries, we build a TR-tree, which is inspired by \cite{wen2015kstr}. TR-tree is built upon IR-tree~\cite{li2011ir} with temporal information. Each leafnode in TR-tree is associated with a spatial object $o$ formed in $\langle o.l, o.c, o.W, o.T \rangle$, in which $o.l$ and $o.c$  is the location and category of $o$, $o.W$ is a set of $o$'s keywords weight calculated by TF-IDF. $o.T$ is the check-in time distribution for those time intervals defined in Sec.~\ref{formulation}. As for non-leafnodes in TR-tree, we extend the strategy of document summary in IR-tree~\cite{li2011ir} to carry spatio-textual and temporal information in its non-leafnode node $N$. For geo-spatial measurement, we add a novel feature, i.e., category entropy, to measure the heterogeneity of venue categories in the area of TR-tree node. In the real world, commercial streets providing various services are more popular. For example, a man may sing karaoke after midnight snack and he may watch movies after shopping. Category entropy is formally defined as follow.

\begin{definition}[Category Entropy]\label{def:4}
  Given a TR-tree node $\eta$, its category entroy is denoted as $\eta.cEntroy$. We denote the set of spatial objects with category $c_i$ in the area of $\eta$ as $S_{c_i}(\eta)$, the entire set of spatial objects in $\eta$ as $S(\eta)$, and the entire set of categories as $Cat$. The category entropy $\eta.cEntroy$ can be defined in Eqn.~\eqref{eq:cat_entropy}.
  \begin{equation} \label{eq:cat_entropy}
  \eta.cEntroy = -\sum_{c_i\in Cat}\frac{\left|S_{c_i}(\eta)\right|}{\left|S(\eta)\right|}\log\frac{\left|S_{c_i}(\eta)\right|}{\left|S(\eta)\right|}
  \end{equation}
  
\end{definition}

Thus, a non-leafnode $N$ is of form $\langle N.rec, N.cEntropy, N.W, N.T \rangle$. $N.rec$ is the minimum bounding rectangle (MBR) of TR-tree. $N.W$ is a TF-IDF weight set containing all the TF-IDF weights of $N$'s descendant nodes' keywords that are calculated by the maximal term frequency TF$^{max}$ and IDF (refer to~\cite{li2011ir} for more details). $N.T$ is the check-in time distribution of $N$ maintaining the maximal check-in distribution for each time interval of its descendant. In general, there are two new elements added in the TR-tree compared with the IR-tree, i.e., the $o.c$ or $N.cEntropy$ for the geo-spatial category or category entropy, $o.T$ or $N.T$ for check-in distribution. Next, we detail the design of $o.T$, $N.T$, and a novel time-aware criterion, visiting time score, for each node.

We separate one day into hourly-based time intervals denoted as $\Gamma$ inspired by~\cite{wen2015kstr}. We define the check-in probability for any node $\eta$ during a time interval $\tau\in\Gamma$ as $\eta.T(\tau)$. Then, for each spatial object $o$ associated in TR-tree's leafnode, we maintain a set of check-in distribution of each time interval, i.e., $o.T = \bigcup_{\tau\in \Gamma}o.T(\tau)$, which is formally defined in Sec.~\ref{formulation}. In order to alleviate the storage overhead, for non-leafnode $N$, we store the check-in probability as the maximum among its descendants':
\begin{displaymath}
N.T = \bigcup_{\tau\in \Gamma}\max_{\eta \in N.children}\eta.T(\tau)
\end{displaymath}

For any node $\eta$, $\eta.T(t)$ is equal to $\eta.T(\tau)$ for  $t\in \tau \in \Gamma$. Thereafter, we define the visiting time score function $F_t$  at time $t$ as:
\begin{equation}
\label{eq:visiting_time_score}
F_t(\eta,t) = \frac{\eta.T(t)}{\max_{\tau\in\Gamma}\eta.T(\tau)}
\end{equation}

\subsection{Social Relationship Network Embedding}

We measure social relationship effect among users with the following steps: neighbor selection method, network embedding strategy, user-inverted storage scheme. 

\subsubsection{Neighbor Selection.}\label{neighbor_selection} In this work, we first take full advantage of users' historical check-in records to extract the following feature of users' preference:

\begin{itemize}
  \item Check-in area: We adopt a spatio-temporal cluster algorithm ST-DBSCAN~\cite{birant2007st} on users' check-in records to obtain a set of clusters $Clr$. For each user $u$, we maintain a vector of the check-in number for each cluster by $\bigcup_{clr_i\in Clr}C(u,clr_i)$, where $C(u,clr_i)$ is the number of $u$'s check-ins in the cluster $clr_i$. 
  \item Check-in time: We map the time of previous check-ins into hourly-based time intervals $\Gamma$. For each user $u$, we maintain a vector of the check-in number for each time interval by $\bigcup_{\tau\in \Gamma}C(u,\tau)$, where $C(u,\tau)$ is the number of $u$'s check-ins during the time interval $\tau$. 
  \item Check-in category: For each user $u$, we maintain a vector of the check-in number for each category by $\bigcup_{c_i \in Cat}C(u,c_i)$, where $C(u,c_i)$ is the number of $u$'s check-ins to the spatial objects with category $c_i$. 
\end{itemize}

For these three features, we use the Cosine similarity metric to calculate the check-in preference similarity between users. After that, given a target user $u$, non-dominated users can be determined by a multi-dimensional optimization Skyline algorithm~\cite{lee2014toward}. In this case, it is said that user $u_i$ dominates another user $u_j$ if $v_i$ is not less than $v_j$ in all dimensions of similarities with $u$ and is better than $v_j$ at least in one dimension ($v_{k}$ is the embedding vector corresponding to $u_{k}$). We propose that non-dominated users are a part of neighbors that will affect the target user most. Besides, for a target user $u$, we merge $u$'s non-dominated users and $u$'s friends to be the final neighbors of $u$.

\subsubsection{Network Embedding.} NETR-Tree requires the similarities between the user and neighbors. As shown in the upper right of Fig.~\ref{fig:overview_NETR-Tree}, on the top of TR-tree we leverage LINE~\cite{tang2015line} to learn a network embedding from users' relationship network structure. Formally, we take the social relationship network $(U,E)$ defined in Sec.~\ref{formulation} as an input graph of LINE, and after training structure features from $(U,E)$, LINE learns a $n\times d$ matrix $\emph{\textbf{V}}$ consisting of all users' represented vectors (Fig.~\ref{fig:overview_NETR-Tree}) where $n$ is the number of user and $d$ is the dimension of the embedding. Thereafter, for a pair of neighbors $u_i$ and $u_j$, we compute the cosine similarity of their corresponding embedding vectors $\emph{\textbf{V}}(u_i)$ and $\emph{\textbf{V}}(u_j)$ to weight the social relationship effect. Then, we define the social relationship effect function $F_s$ for user $u$'s decision of visiting NETR-Tree leafnode $o$ as:

\begin{equation} 
\begin{split}
  \label{eq:social_friendship_effect_object}
  F_s(u,o) = \frac{1}{|u.Nrs|} \sum_{u_i \in u.Nrs} \frac{\emph{\textbf{V}}(u_i)\cdot \emph{\textbf{V}}(u)}{|\emph{\textbf{V}}(u_i)||\emph{\textbf{V}}(u)|} \\ \times \frac{C(u_i,o)}{\max_{o_j \in \{o\text{'s brothers\}}}C(u_i,o_j)}
\end{split}
\end{equation}
where $u.Nrs$ is a set of $u$'s neighbors and $C(u_i,o_j)$ is the number of $u_i$'s historical check-ins in $o_j$. In Eqn.~\eqref{eq:social_friendship_effect_object}, $F_s$ is normalized by the number of $u$'s neighbors and each neighbor's maximal check-in number within $o$'s brother nodes in NETR-Tree. Therefore, in the internal nodes of NETR-Tree, for each user $u$, we should get access to $u$'s historical check-ins records. To optimize both processing time and space consumption, we present the user-inverted storage scheme.

\subsubsection{User-inverted Storage.}\label{scheme:user-inverted} The upper right of Fig.~\ref{fig:overview_NETR-Tree} illustrates our \emph{user-inverted storage scheme}, i.e., network embedding user-check-in value blocks. The blocks are maintained based on TR-tree and it has two parts, namely, an index block and value blocks. Similar with inverted file, users are the user-check-in file's equivalent of keywords. Consequently, the index block consists $|U|$ entries. Each entry for user $u$ contains its corresponding embedding vector $\emph{\textbf{v}}$, and points to a value block that contains user-check-in values. Inside a value block of user $u$ is a list of $\{\eta,C(u,\eta)\}$ recording the number of $u$'s check-ins in TR-tree node $\eta$. Similar with temporal check-in information, to cut down space redundancy, for non-leafnode $N$, $C(u,N)$ records $u$'s maximum check-in number among all objects in $N.rec$.

\begin{algorithm}
\SetKwFunction{LINE}{LINE}
\caption{\small{Network Embedding User-check-in Value Blocks Construction}}
\KwIn{social relationship network, $(U,E)$; \\ \ \ \ \ \ \ \ \ \ the leafnodes set of TR-tree, $O$; \\ \ \ \ \ \ \ \ \ \ the list of visited objects for each user, $L$; \\ \ \ \ \ \ \ \ \ \ the users' check-in history in objects, $C$ }
\KwOut{network embedding user-check-in value blocks, $NEB$}
\label{alg:network_embedding_construction}
$\bm{V} \leftarrow $\LINE$(U,E) ${\tcp*[l]{obtain a network embedding matrix}}
$NEB \leftarrow \emptyset${\tcp*[l]{declare network embedding blocks}}
\ForEach{user $u_i \in U$}
{
  $\bm{v_i} \leftarrow \bm{V}(u_i)$\;
  $VB(u_i) \leftarrow \emptyset${\tcp*[l]{declare a value block for $u_i$}}
  \ForEach{object $o_i \in L(u_i)$}
  {
    $VB(u_i) \leftarrow VB(u_i) \cup \{o_i, C(u_i,o_i)\}$\;
    $\eta \leftarrow O(o_i)$\;
    \While {$\eta.parent$ is not null}
    {
      \If{$VB(u_i)$ contains no check-in history \ \ \ of $u_i$ in $\eta.parent$} 
      {
        $VB(u_i) \leftarrow VB(u_i) \cup \{\eta.parent,C(u_i,\eta)\} $\;
      }
      \ElseIf{$C(u_i,\eta.parent) < C(u_i,\eta)$} 
      {
        $C(u_i,\eta.parent) \leftarrow C(u_i,\eta)$\;
      }
      $\eta \leftarrow \eta.parent$\;
    }
  }
  $NEB \leftarrow NEB \cup \{\bm{v_i}, VB(u_i)\}$\;
}
\Return{$NEB$}\;
\end{algorithm}

Algorithm~\ref{alg:network_embedding_construction} outlines the implementation of network embedding and user-inverted storage scheme. After obtaining a network embedding matrix of users by LINE (line 1), for each user $u$, a bottom-up update strategy, maintaining user-check-in values from leafnodes up to the root, is applied (lines 6-14) to improve the construction efficiency. Thereafter, we define the social relationship effect function $F_s$ for user $u$'s decision at any node $\eta$ as:

\begin{equation}
\begin{split}
  \label{eq:social_friendship_effect}
  F_s(u,\eta) = \frac{1}{|u.Nrs|} \sum_{u_i \in u.Nrs} \frac{\emph{\textbf{V}}(u_i)\cdot \emph{\textbf{V}}(u)}{|\emph{\textbf{V}}(u_i)||\emph{\textbf{V}}(u)|} \\\times \frac{C(u_i,\eta)}{\max_{\eta_j \in \{\eta\text{'s brothers\}}}C(u_i,\eta_j)}
\end{split}
\end{equation}

\section{Query Processing} \label{query_processing}

In this section, we present STSKQ processing algorithm based on NETR-Tree and the analysis of its time complexity.

\subsection{STSKQ Using NETR-Tree}

To process STSKQ returning a set $Top_k(u,q)$ for user $u$ and query $q$, we exploit the best-first traversal that searches the entry with the largest score in a heap. The score function includes geo-spatial score, keywords similarity, visiting time score, and social relationship effect, while visiting time score and social relationship effect are defined in Eqn.~\eqref{eq:visiting_time_score} and Eqn.~\eqref{eq:social_friendship_effect}. Thus, we define the score for geo-spatial score, keywords similarity as follows:

\begin{definition}[Geo-spatial Score]\label{def:location_proximity}
Geo-spatial scores are comprised of two modules: category entropy and location proximity. Category entropy is defined in Eqn.\eqref{eq:cat_entropy}. Let $\delta_{max}$ denote the maximal search radius in the location layer, $\delta(q,o)$ be the Euclidian distance between query $q$ and leafnode, i.e., spatial object $o$, and $\min\delta(q,N.rec)$ represent the minimum Euclidian distance between $q$ and non-leafnode $N$'s MBR. The location proximity between $q$ and NETR-Tree node $\eta$ is defined as:
\begin{equation}
  \label{eq:location_proximity}
  l(q,\eta) = 
  \begin{cases}
    1-\frac{\delta(q,\eta)}{\delta_{max}} &   \text{$\eta$ is a leafnode}   \vspace{0.1in}\\
    1-\frac{\min\delta(q,\eta.rec)}{\delta_{max}} & \text{$\eta$ is a non-leafnode} 
  \end{cases}
\end{equation}
Thereafter, the geo-spatial score between $q$ and NETR-Tree node $\eta$ is defined as:
\begin{equation}
\label{eq:geo_spatial_score}
F_g(q,\eta) = \theta \times \eta.cEntropy + (1 - \theta) \times l(q,\eta)
\end{equation}

\end{definition}

\begin{definition}[Keywords Similarity]\label{def:keywords_similarity}
As stated in Sec.~\ref{NETR}, $o.W$ and $N.W$ are the sets that contain all keywords' TF-IDF weight of $o$ and $N$ respectively. Therefore, for any node $\eta$ in NETR-Tree, the keywords similarity between $q$ and $\eta$ is defined as:
\begin{equation}
  \label{eq:keywords_similarity}
  F_k(q,\eta) = \frac{1}{\phi_{max}\times|q.W|} \sum_{w\in q.W}\eta.W(w)
\end{equation}
where $\phi_{max}$ is used for normalization.
\end{definition}

As demonstrated in~\cite{vi2012spatial}, a complicated ranking score function is necessary when we take temporal and social information into account in spatial keyword query. Finally, combining Eqn.~\eqref{eq:visiting_time_score}, Eqn.~\eqref{eq:social_friendship_effect}, Eqn.~\eqref{eq:geo_spatial_score}, and Eqn.~\eqref{eq:keywords_similarity}, a carefully designed ranking score function for an node $\eta$ in NETR-Tree is defined as:
\begin{equation}
\begin{split}
  \label{eq:ranking_score_objects}
  F(u,q,\eta) = \ & \alpha\times F_g(q,\eta) + \beta\times F_k(q,\eta) + \gamma \times F_s(u,\eta)  \\ & + (1-\alpha-\beta-\gamma) \times  F_t(\eta,q.t)
\end{split}
\end{equation}

STSKQ processing is sketched in Alg.~\ref{alg:netr_search}. A max heap is employed to keep the index nodes and objects sorted in descending order of their scores (line 1). If the first entry in the heap is an object, it is the best object in the heap and will be inserted into the result set $Top_k(u,q)$ (lines 3-5). For any node $\eta$ at time $t$, if $\eta.T(t)$ is $0$, it indicates that all the objects inside $\eta$'s area are closed at $t$, and so it is unnecessary to visit $\eta$'s child/descendant nodes (lines 6-7). Besides, objects and nodes outside search radius, not containing all query keywords or with smaller scores
than the top-$k$ objects in the heap are pruned out (lines 6-7,10). In the end, it will return top-$k$ objects.

\begin{algorithm}
\SetKw{Continue}{continue}
\caption{\small STSKQ Using NETR-Tree}
\label{alg:netr_search}
\KwIn{a user, $u$; a query, $q$; Top-$k$ result, $k$; the root of NETR-Tree, $root$ }
\KwOut{Top-$k$ objects, $Top_k(u,q)$}
Maxheap.insert($root$, $\infty$)\;
\While {Maxheap.size() $\neq$ $0$}
{
  $N$ $\leftarrow$ Maxheap.first()\;
  \uIf{ $N$ is an object}
  {$Top_k(u,q)$.insert($N$)\;
  }
  \uElseIf{$N.T(q.t)=0$ or $q.W \nsubseteq N.W$ or $\delta(N,q)>r$}
  {
    \Continue\;
  }
  \Else{
    \For {$n_i$ $\in$ $N$.entry}
    {
      \If{Number of objects with larger score \ \ \ than $ F(u,q,n_i)$  in Maxheap  $<$ ($k-Top_k(u,q)$.size())}
      {Maxheap.insert($n_i$, $ F(u,q,n_i)$)\;}
    }
  }
  
}
\Return{$Top_k(u,q)$}\;
\end{algorithm}

We prove the correctness of Alg. \ref{alg:netr_search} by Thm. \ref{theorem:1}, which guarantees the theoretical reliability of NETR-Tree.
\begin{theorem}\label{theorem:1}
Given a user $u$, the score of an internal node $N$ is larger than its descendant object $o$ for any query $q$.
\end{theorem}
\begin{proof}
First, for an internal node $N$, the MBR of $N$ encloses all descendant objects, i.e., $\forall o \in$ $N$'s descendants, $\min\delta(q,N.rec) \le \delta(q,o)$, and the categories heterogeneity of $N$ must be not less than it descendants. Hence, it follows $F_g(q,N) \ge F_g(q,o)$. Second, since TF-IDF weight is the multiplication of IDF and TF$^{max}$ in $N$, i.e., $\max_{d\in D_N}(tf_{w,d})$ where $D_N$ represents all the text documents for objects inside $N$, it indicates that $F_k(q,N) \ge F_k(q,o)$. Finally, since for node $N$, both check-in probability and user check-in number are maximal among $N$'s descendants, we have $F_t(N,t) \ge F_t(o,t)$ and $F_s(u,N) \ge F_s(u,o)$ at any query time $t$, for any user $u$. All these inequalities lead to $F(u,q,N) \ge F(u,q,o)$.
\end{proof}


\subsection{Theoretical Analysis}
The searching efficiency depends on the number of candidate nodes. Those searching algorithm evaluating objects according to single criterion sequentially is obviously inefficient. Nevertheless, our algorithm can perform (1) spatial pruning, (2) textual filtering, and (3) temporal check at the same time. Next, we will give an analysis of the time complexity of our algorithm.~\cite{chen2017time} takes the product of the number of objects accessed and average processing time of single object as the time complexity of their algorithm. We adopt the same estimation method.

Evidently, an object $o$ with high possibility of being a candidate for a query $q$ should satisfy: (1) $o.W \supseteq q.W$, (2) $o.T(q.t)\neq 0$, and (3) $\delta(o.l,q.l) \leq$ the search radius $r$. For these three requirements, we defined contain-keyword probability $P(o.W \supseteq q.W)$, exist-previous-record probability $P(o.T(q.t) \neq 0)$, and within-radius $P(\delta(o.l,q.l)\leq r)$. Let $P_{can}(o)$ be the probability of the event that $o$ is a candidate and $P_{can}(o) = P(o.W \supseteq q.W) \cdot P(o.T(q.t) \neq 0) \cdot P(\delta(o.l,q.l) \leq r)$. Then we present how to calculate these 3 probabilities respectively.

\subsubsection{Contain-Keyword Probability.}\label{ckp}
As~\cite{wu2012joint} has pointed out, the keywords in the data set follow a Zipfian distribution. $w_{i}$ is the word with the $i$-th highest occurrence-frequency. The occurrence probability of $w_{i}$ $P_{oc}(w_i)$ can be calculated through: 
$$P_{oc}(w_{i}) = \frac{i^{-s}}{\sum_{j=1}^{\vert D.w \vert}w_{j}^{-s}}$$
where $\vert D.w \vert$ is the total number of words in the data set $D$ and $s$ characterizes the skewness. Using the estimation model of Wu et al.~\cite{wu2012joint}, the contain-keyword probability is defined as
\begin{equation}
\begin{split}
  \label{ts1}
  P(o.W \supseteq q.W) = \sum_{any~list~of~o.key}\prod_{j=1}^{\vert q.W \vert}\frac{P_{oc}(w_{q.W_j})}{1-\sum_{k=1}^{j-1}P_{oc}(w_{q.W_{k}})}
\end{split}
\end{equation}

\subsubsection{Exist-Previous-Record Probability.}\label{eprp}
There exists a pattern in people's daily life that people tend to do something at one certain moment intensively, which follows a normal distribution. Intuitively, the number of check-ins of an object $o$ also follows this pattern. In practice, there is always more than one peak in the distribution. Hence, we use Gaussian Mixture Models. The probability of check-in occurring at a certain moment $t$ follows a superposition of a sequence of normal distributions $N_{1}(\mu_{1},\sigma_{1}), N_{2}(\mu_{2},\sigma_{2}),\ldots, N_{k}(\mu_{k},\sigma_{k})$. Then, we can get 
\begin{equation}
\begin{split}
  \label{ts2}
  P(o.T(q.t) \neq 0) = \sum_{i=1}^{k}\int_{on~\tau}f_i(t)dt
\end{split}
\end{equation}
where $\tau \ni q.t$ and $$f_{i}(t) = \frac{1}{\sqrt{2\pi\sigma_{i}}}\exp(-\frac{(x-\mu_{i})^2}{2\sigma_i})$$.

\subsubsection{Within-Radius Probability.}\label{wrp}
Pruning via spatial constraints is fundamental in spatial keyword query. Chen et al.~\cite{chen2017time} just assume the query radius is $\infty$ ignoring the spatial constraints to simplify the computation of time complexity. Here, we give an approximate estimation of within-distance probability $P(\delta(o.l,q.l) \leq r)$. We consider a square whose side-length is $\sqrt{\pi}\cdot r$, thus having the same area as $\odot(q.l,r)$. Leafnodes having intersection with this square have the potential objects that can be candidates. We aim to find the minimum number of leafnodes' MBRs that can cover the square. Let all the leafnodes' MBRs be placed in a rectangular plane coordinate system. The shortest side-length of all MBRs in $x$-axis direction is defined as $l_{x}$, and the shortest distance in the direction of $x$-axis between any two MBRs that do not overlap in $x$-axis direction is defined as $d_{x}$. We define $l_{y}$, $d_{y}$ in a similar way. Assuming the number of objects contained in a MBR is a constant $K$, the within-distance probability is defined as
\begin{equation}
\begin{split}
  \label{ts3}
  P(\delta(o.l,q.l) \leq r) \approx \frac{K}{\vert D \vert} \times\lceil\frac{\sqrt{\pi}\cdot r}{l_{x}+d_{x}}\rceil \times\lceil\frac{\sqrt{\pi}\cdot r}{l_{y}+d_{y}}\rceil
\end{split}
\end{equation}

With Eqn.~\eqref{ts1}~\eqref{ts2}~\eqref{ts3}, we can figure out the value of $P_{can}(o)$. Searching from the root of tree to a leafnode takes $O(\lg\vert D\vert)$ time. Assuming the time cost of grading a object $o$ is a constant $\zeta$, the average processing time of single object is $O(\lg\vert D\vert + \zeta)$. Hence, the time complexity of our algorithm is
\begin{equation}
\begin{split}
  \label{ts4}
  O(\vert D\vert\cdot P_{can}(o)\cdot (\lg\vert D\vert) + \zeta))
\end{split}
\end{equation}

\section{Experiments}\label{Experiment}

We systematically evaluate the performance of our proposed index and algorithms compared with state-of-the-art methods on two large-scale real-world datasets here. All the indices and algorithms are implemented in Python and run on a Linux server with 2.1 GHz Intel Xeon processor and 64GB RAM. All experiments are repeated 5 times and averaged results are reported. 

\subsection{Baseline}\label{baseline}
To give a comprehensive comparison, we implement three baseline frameworks, including one representative baseline framework IR-tree~\cite{li2011ir}, a recent state-of-the-art tree-based framework Routing R-tree~\cite{liu2018happened}, and a non-tree-based framework SKB-Inv index~\cite{zhang2018augmented}. Notice that these three baseline algorithms cannot solve STSKQ directly and demand for enhancing methods. These indices and methods are listed as follows:

\begin{itemize}
   
\item\textbf{IR-tree:} An IR-tree is an R-tree extended with inverted files. To tackle STSKQ, IR-tree firstly retrieves a candidate of objects by location proximity and keywords similarity. Then we rank the candidate with social relationship effect $F_s$ in Eqn.~\eqref{eq:social_friendship_effect} and visiting time score $F_t$ in Eqn.~\eqref{eq:visiting_time_score} to return the top-$k$ objects.

\item\textbf{Routing R-tree:} Routing R-tree enables spatio-temporal keyword search by constructing an R-tree for each time interval. Hence, to deal with STSKQ, Routing R-tree leverages time segment scheme to construct a group of R-trees. Each R-tree is built on a corresponding time interval in $\Gamma$, and we process query in one of the R-trees according to user's query time. Similarly, we further select the top-$k$ objects in accordance with $F_s$ in Eqn.~\eqref{eq:social_friendship_effect}.

\item\textbf{SKB-Inv index:} SKB-Inv index adopts k-means clustering algorithm in order to group objects by their spatial attribute, and further organizes the spatial objects into inverted lists based on other attributes or keywords. To process STSKQ, we regard check-in probabilities for the $24$ time intervals and users' check-in records as additional attributes for the group of objects. Thereafter, we retrieve the top-$k$ objects combining all the inverted lists in Eqn.~\eqref{eq:ranking_score_objects}.

\end{itemize}

\subsection{Experimental Setup}

\begin{table}
\centering
\caption{Datasets Statistics}
\label{tab:2}
      \begin{tabular}{ccc}
        \hline
        \textbf{Dataset} & \textbf{Yelp} & \textbf{Weeplaces} \\
        \hline
        \#objects & $99,798$ & $99,378$\\
        \#check-ins & $15,816,233$ & $7,658,368$ \\
        \#users & $527,532$ & $16,021$  \\
        \#neighbors & $16.7$ & $7.5$ \\
        \hline
    \end{tabular}
\end{table}

We use two real-world datasets, Yelp\footnote{Available at https://www.yelp.com/dataset} and Weeplaces, to verify the effectiveness of NETR-Tree. Yelp is obtained from Yelp Dataset Challenge and Weeplaces~\cite{liu2014exploiting} is collected from the popular LBSN Weeplaces. Both datasets contain geographic locations, keywords, check-in time, and relationship information. Table~\ref{tab:2} reports the statistical information of the two datasets.

We investigate the performance of our proposed index and algorithms for STSKQ under sorts of parameters listed in Table~\ref{tab:3}. Besides, both $\alpha$ and $\beta$ in Eqn.~\eqref{eq:ranking_score_objects} are all set to $0.25$, $\theta$ in Eqn.~\eqref{eq:geo_spatial_score} is set to $0.5$. For every set of the experiments, $100$ random queries are evaluated to measure both the average processing time and the average disk I/O (i.e., the number of node accesses). 

\begin{table}
\centering
\caption{Parameter Setting}
\label{tab:3}
\begin{tabular}{ccc}
      \hline
      \textbf{Parameter} & \textbf{Range} & \textbf{Default} \\
      \hline
      k  & $1,3,5,7,9$ & $5$ \\
      $|$q.W$|$ & $1,3,5,7,9$ & $5$\\
      search radius (km) & $4,8,12,16,20$ & $12$ \\
      $\gamma$ in Eqn.~\eqref{eq:ranking_score_objects}  & $0.1,0.2,0.3,0.4,0.5$ & $0.3$ \\
      \hline
\end{tabular}
\end{table}

\subsection{Performance Evaluation}

In this section, we conduct a set of experiments to evaluate our NETR-Tree on the efficiency of index construction and STSKQ processing in two different real-life LBSNs datasets, compared with different baseline algorithms proposed in Sec~\ref{baseline}. We evaluate the efficiency of STSKQ processing the query from several aspects: the number of the result objects, the number of query keywords, the influence of different search radii and the weight of the social relationship effect.

{\bf Index construction cost: }
We first evaluate the construction time and index size of NETR-Tree with three baseline algorithms on two datasets Yelp and Weeplaces in Fig.~\ref{fig:construction_cost}. Particularly, Routing R-tree and SKB-Inv index have a dramatically high cost in time and space for both datasets as shown in  Fig.~\ref{fig:construction_cost}(a)(b). Routing R-tree maintains an R-tree for every time-interval while an object can exist in many different R-trees simultaneously. For SKB-Inv index, it maintains an inverted list for every time interval, keyword, and user's check-in record. Both of these methods lead to extremely high redundancy. In contrast, as demonstrated in Fig.~\ref{fig:construction_time}, IR-tree is the most constructing-efficient due to the absence of temporal and social information. Despite taking all the information into account, NETR-Tree is the most efficient and lightweight in construction cost if we exclude IR-tree.

{\bf Effect of $k$: }
Next, we investigate the effect of varying $k$ (i.e., the number of the objects returned) on processing time and I/O cost. With the result shown in Fig.~\ref{fig:effect_k}, NETR-Tree exceeds other algorithms by a wide margin. Besides, all the algorithms perform better in Weeplaces than Yelp, especially for Routing R-tree and SKB-Inv index. The reason is that Weeplaces has much less social relationship information than Yelp, which alleviates the load in processing social information. IR-tree has the worst performance on processing time in Fig.~\ref{fig:effect_k}(c)(d). Without proper processing of temporal and social information, IR-tree retrieves large numbers of false positive objects in its candidate, and the additive operation to filter those objects leads to high time cost.

\begin{figure}
  \centering
  \subfigure[Construction Time]{
    \label{fig:construction_time}  
    \includegraphics[width=0.23\textwidth]{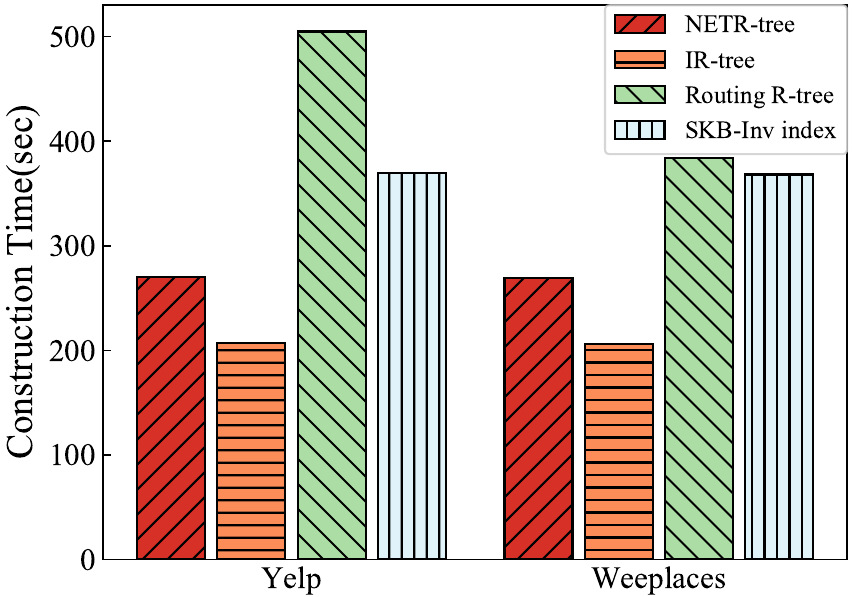}}
  \subfigure[Index Size]{
    \label{fig:index_size}   
    \includegraphics[width=0.23\textwidth]{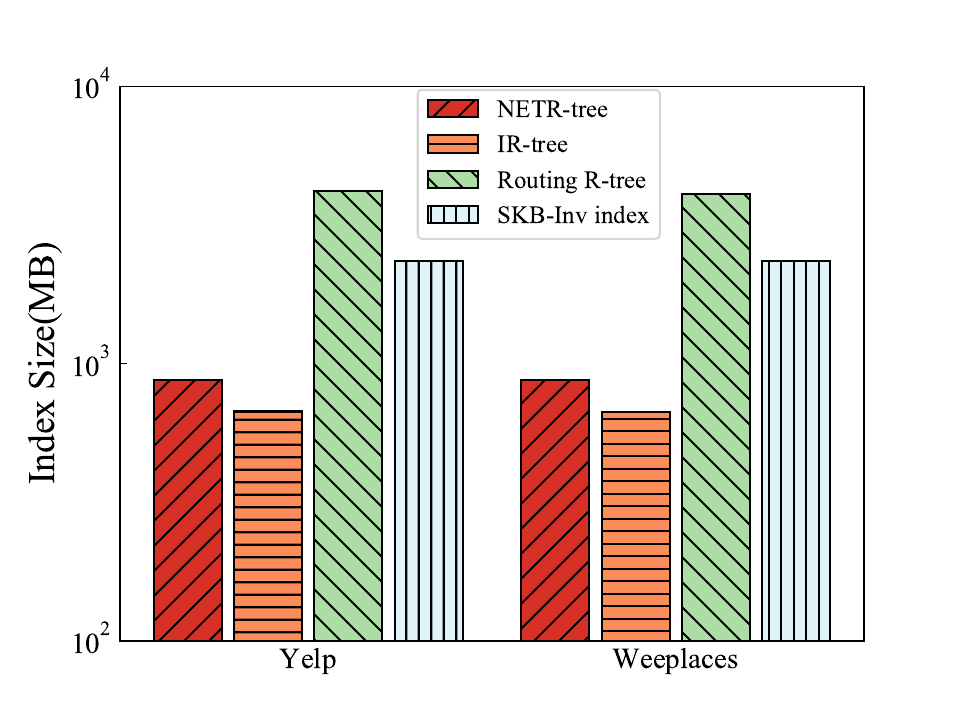}}
  \caption{Index construction cost}
  \label{fig:construction_cost} 
\end{figure}

{\bf Effect of $|$q.W$|$: }
Then, we study the impact of the number of query keywords, as illustrated in Fig.~\ref{fig:effect_keywords}. Clearly, NETR-Tree performs better than other algorithms in both processing time and I/O cost. In addition, it can be seen from Fig.~\ref{fig:effect_keywords}(c)(d) that the processing cost of SKB-Inv index ascends with the growth of $|$q.W$|$, since  SKB-Inv index needs to scan more group of objects with more keywords requested. Furthermore, as for processing time, Routing R-tree performs well in Weeplaces as shown in Fig.~\ref{fig:effect_keywords}(c), whereas it has a poor performance in Yelp in Fig.~\ref{fig:effect_keywords}(a). Similarly, the reason for the latter is that Yelp has much heavier burden in processing social relationship information. In Weeplaces, the good performance of Routing R-tree as well as NETR-Tree is mainly because of their time segment scheme and temporal information. In return, Routing R-tree makes sacrifices on its construction cost, while NETR-Tree shows an effective trade off between construction and processing.

\begin{figure*}
  \includegraphics[width=\textwidth]{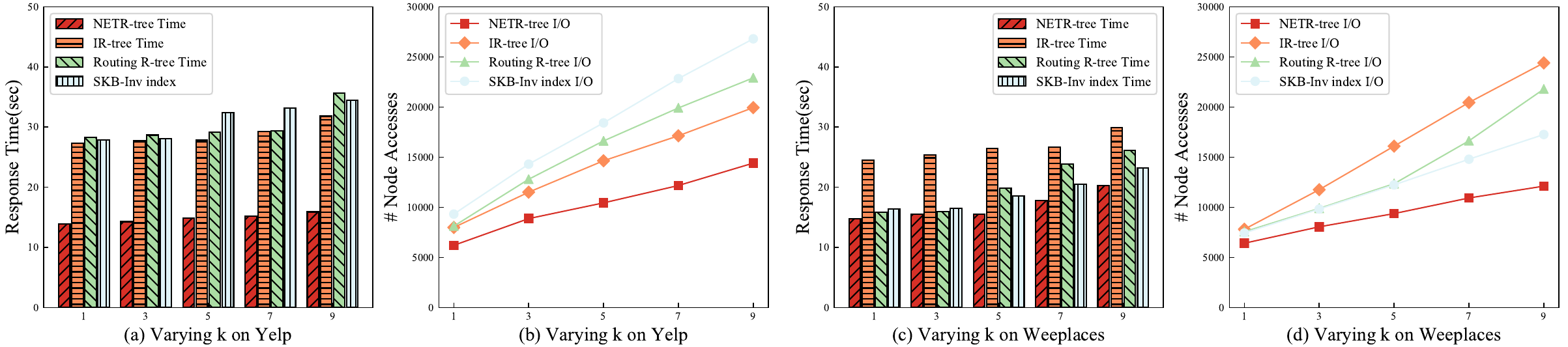}
  \caption{Effect of k} \label{fig:effect_k}
\end{figure*}

\begin{figure*}
  \includegraphics[width=\textwidth]{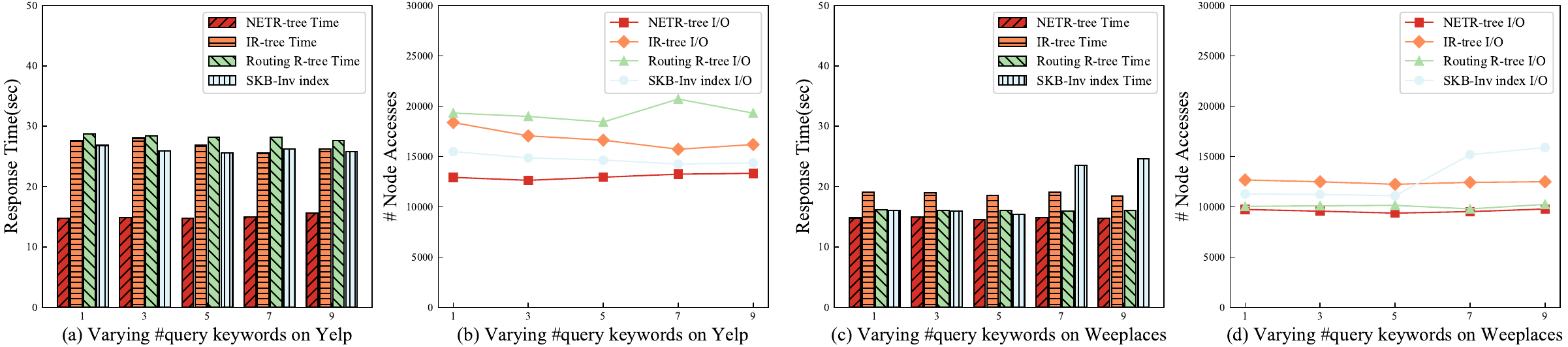}
  \caption{Effect of $|$q.W$|$} \label{fig:effect_keywords}
\end{figure*}

\begin{figure*}
  \includegraphics[width=\textwidth]{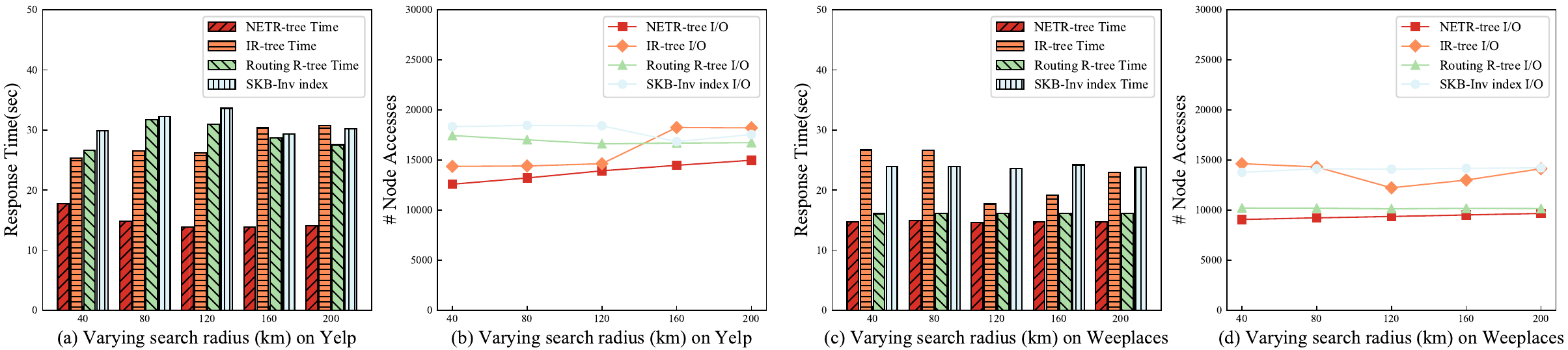}
  \caption{Effect of search radius} \label{fig:effect_area}
\end{figure*}

\begin{figure*}
  \includegraphics[width=\textwidth]{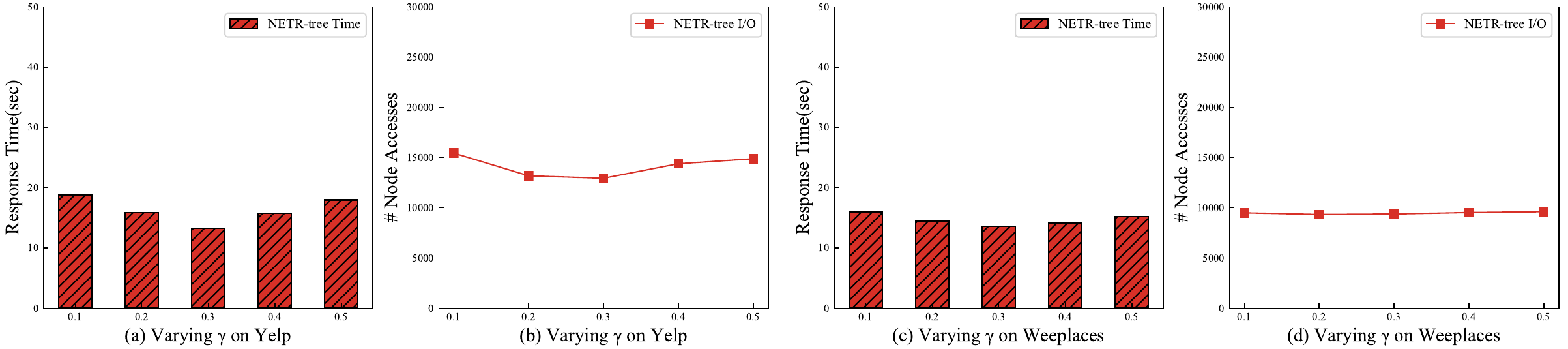}
  \caption{Effect of $\gamma$} \label{fig:effect_gamma}
\end{figure*}

{\bf Effect of search radius (km): }
In this experiment, we evaluate the influence of different search radii. As depicted in Fig.~\ref{fig:effect_area}, NETR-Tree again performs the best since the temporal and social information helps prune out many irrelevant objects. Moreover, as shown in Fig.~\ref{fig:effect_area}(c)(d), owing to time segment scheme, the processing time of NETR-Tree and Routing R-tree stay stable while other algorithms' performance varies with search radius.

{\bf Effect of $\gamma$: }
Last but not least, we inspect the effect of $\gamma$ in Eqn.~\eqref{eq:ranking_score_objects}, where $\gamma$ is the weight of the social relationship effect. Since Eqn.~\eqref{eq:ranking_score_objects} is only used for our NETR-Tree, this is an internal experimental evaluation. As shown in Fig.~\ref{fig:effect_gamma}(a), with the growth of $\gamma$, the processing time of NETR-Tree in Yelp decreases at first and then increases after $\gamma$ reaches $0.3$. The reason is that with the weight of social relationship effect increasing, NETR-Tree can prune out more socially irrelevant objects. In the meantime, other criteria are losing their weight in the score function, which leads to the subsequent upswing. Moreover, the better performance in Fig.~\ref{fig:effect_gamma}(c)(d) compared with Fig.~\ref{fig:effect_gamma}(a)(b) demonstrates again that it is more difficult to process a query in Yelp with more social information. This further exemplifies the effective application of our proposed NETR-Tree.

\section{Conclusion} \label{conclusion}
In this paper, we formulate the \emph{Social-based Time-aware Spatial Keyword Query} (STSKQ), which takes both spatial constraint, temporal information, and social relationship into consideration. To address it, we propose a novel index structure, which is a two-layer hybrid framework named Network Embedding Time-aware R-tree (NETR-Tree). The two-layer scheme and query processing algorithms are designed to tackle STSKQ efficiently. In order to deal with massive user relationship networks, NETR-Tree exploits the network embedding strategy to measure the social effect when a user issues a query. Thus, NETR-Tree can give top-$k$ result objects based on geo-spatial score, keywords similarity, visiting time score, and social relationship effect. Finally, extensive experiments on two real datasets verify the efficiency and effectiveness of NETR-Tree.


\bibliographystyle{IEEEtran}
\bibliography{NETR-Tree}

\end{document}